\documentclass[english]{llncs}
\usepackage[T1]{fontenc}
\usepackage[latin9]{inputenc}
\usepackage{verbatim}
\usepackage{float}
\usepackage{amsmath}
\usepackage{amssymb}
\usepackage{esint}
\usepackage[normalem]{ulem}
\usepackage{graphicx}

\usepackage{geometry} 
\usepackage{algorithmic}
\usepackage[linesnumbered, ruled, vlined, norelsize]{algorithm2e}

\usepackage{color}

\def\DEBUG{true}

\ifdefined\DEBUG
\definecolor{marekgreen}{RGB}{0,185,0}

\newcommand{\fab}[1]{\textcolor{red}{#1}}
  \def\rem#1{{\marginpar{\raggedright\scriptsize #1}}}
  \newcommand{\fabr}[1]{\rem{\textcolor{red}{$\bullet$ #1}}}
  \newcommand{\joyr}[1]{\rem{\textcolor{blue}{$\bullet$ #1}}}
  \newcommand{\marr}[1]{\rem{\textcolor{marekgreen}{$\bullet$ #1}}}
\else

  \newcommand{\fab}[1]{#1}

  \newcommand{\fabr}[1]{}
  \newcommand{\joyr}[1]{}
  \newcommand{\marr}[1]{}
\fi

\newcommand{\Neigh}[1]{\delta(#1)}
\newcommand{\GKPS}{\operatorname{GKPS}}

\newcommand{\hybrid}{\operatorname{Hybrid}}
\newcommand{\MAIN}{\operatorname{MAIN}}

\newcommand{\Gre}{\operatorname{Greedy}}

\usepackage{xspace}
\newcommand{\euler}{\mathbf{\mathsf{e}}}
\newcommand{\cI}{\mathcal{I}}

\def \OPT {\ensuremath{\operatorname{OPT}}\xspace}
\def \LP {\ensuremath{\operatorname{LP}}\xspace}
\def \ALG {\ensuremath{\operatorname{ALG}}\xspace}
\newcommand{\xhdr}[1]{\vspace{2mm} \noindent{\textbf{#1}}}
\newcommand{\ie}{{\emph{i.e.,}~\xspace}}
\newcommand{\eg}{{\emph{e.g.,}~\xspace}}
\newcommand{\etal}{{\emph{et al.}~\xspace}}
\usepackage{hyperref}

\usepackage{babel}

\makeatother

\usepackage{babel}
\begin{document}
\global\long\def\adj{\mbox{\footnotesize Adj}}
 \global\long\def\chr#1{\mathbf{1}_{#1}}
 \global\long\def\br#1{\left( #1 \right)}
 \global\long\def\brq#1{\left[ #1 \right]}
 \global\long\def\brw#1{\left\{  #1\right\}  }
 \global\long\def\cut#1{\partial#1 }
 \global\long\def\excond#1#2{\mathbb{E}\left[\left. #1 \right\vert #2 \right]}
 \global\long\def\ex#1{\mathbb{E}\left[#1\right]}
 \global\long\def\E{\mathbb{E}}
 \global\long\def\exls#1#2{\mathbb{E}_{#1}\left[#2\right]}
 \global\long\def\prcond#1#2{\Pr\left[\left. #1 \right\vert #2 \right]}
 \global\long\def\setst#1#2{\left\{  #1\left|#2\right.\right\}  }
 \global\long\def\setstcol#1#2{\left\{  #1:#2\right\}  }
 \global\long\def\set#1{\left\{  #1\right\}  }
 \global\long\def\adj#1{\delta\br{#1}}
 \global\long\def\setst#1#2{\left\{  \left.#1\right|#2\right\}  }
 \global\long\def\set#1{\left\{  #1\right\}  }
 \global\long\def\ind#1{\mathbf{1}\left[ #1 \right]}
 \global\long\def\st#1{[#1] }
 \global\long\def\opstyle#1{\mathbb{#1}}
 \global\long\def\size#1{\left|#1\right|}
 \global\long\def\setstcol#1#2{\left\{  #1:#2\right\}  }
 \global\long\def\set#1{\left\{  #1\right\}  }
 \global\long\def\indi#1{\chi\brq{#1}}
 \global\long\def\evalat#1#2{ #1 \Big|_{#2}}
 \global\long\def\prls#1#2{\opstyle P_{#1}\left[ #2 \right]}
 \global\long\def\excondls#1#2#3{\mathbb{E}_{#1}\brq{\left.#2\right|#3}}
 \global\long\def\prcondls#1#2#3{\Pr_{#1}\brq{\left.#2\right|#3}}
 \global\long\def\st#1{[#1] }
 \global\long\def\indi#1{\chi\brq{#1}}
 \global\long\def\evalat#1#2{ #1 \Big|_{#2}}
 \global\long\def\Df#1#2{\frac{\partial#1}{\partial#2}}
 \global\long\def\hx#1{\hat{x}_{#1}}
 \global\long\def\eps{\varepsilon}
 \global\long\def\M{{\cal M}}
 \global\long\def\pr#1{\Pr\brq{#1}}
 \global\long\def\opstyle#1{\mathbb{#1}}
 \global\long\def\ex#1{\mathbb{E}\left[#1\right]}
 \global\long\def\xp#1{\mathbb{E}\left[#1\right]}
 \global\long\def\prcond#1#2{\Pr \left[\left. #1 \right\vert #2 \right]}
 \global\long\def\excond#1#2{\opstyle E \left[#1 \left|#2\right. \right]}
 \global\long\def\exls#1#2{\opstyle{\opstyle E}_{#1}\left[ #2 \right]}
 \global\long\def\prls#1#2{\opstyle P_{#1}\left[ #2 \right]}
 \global\long\def\br#1{\left( #1 \right)}
 \global\long\def\brq#1{\left[ #1 \right]}
 \global\long\def\brw#1{\left\{  #1\right\}  }
 \global\long\def\size#1{\left|#1\right|}
 \global\long\def\setst#1#2{\left\{  #1\left|#2\right.\right\}  }
 \global\long\def\setstcol#1#2{\left\{  #1:#2\right\}  }
 \global\long\def\pxE{\br{p_{e}x_{e}}_{e\in E}}
 \global\long\def\evalat#1#2{ #1 \Big|_{#2}}
 \global\long\def\X#1{\hat{X}_{#1}}
 \global\long\def\E{\hat{E}}
 \global\long\def\Frac#1#2{#1\left/\br{#2}\right.}
 \global\long\def\P#1{{\cal P}\br{#1}}
 \global\long\def\I{{\cal I}}
 \global\long\def\brqbb#1{\llbracket#1\rrbracket}
 \global\long\def\h#1{\hat{#1}}
 \global\long\def\lrg#1{#1_{large}}
 \global\long\def\sml#1{#1_{small}}

\pagestyle{headings}

\author{ Marek Adamczyk\inst{1} \and Brian Brubach \inst{2} \and Fabrizio Grandoni\inst{3} \and Karthik A. Sankararaman \inst{4} \and Aravind Srinivasan \inst{5} \and Pan Xu \inst{6}}

\institute{Institute of Informatics, University of Warsaw, Poland, \email{m.adamczyk@mimuw.edu.pl} \and
Wellesley College, Wellesley, MA, \email{bb100@wellesley.edu} \and 
IDSIA, University of Lugano, Switzerland, \email{fabrizio@idsia.ch} \and 
Facebook, Menlo Park, CA, \email{karthikabinavs@gmail.com} \and 
University of Maryland, College Park, MD, \email{asriniv1@umd.edu} \and 
New Jersey Institute of Technology, New Jersey, NJ, \email{pxu@njit.edu}
 }

\title{Improved Approximation Algorithms\\ for Stochastic-Matching Problems\thanks{This work was partially done while Adamczyk was visiting IDSIA. Adamczyk and Grandoni were partially supported by the ERC StG project NEWNET no.~279352, the ERC StG project PAAl no.~259515, the ISJRP project Mathematical Programming in Parameterized Algorithms, and the SNSF Excellence Grant 200020B\_182865/1.
Adamczyk was partially supported by ERC StG project TOTAL no.~677651. 
%\fab{Adamczyk} and \fab{Grandoni} were partially supported by the ERC StG project NEWNET no.~279352, and the first author by the ERC StG project PAAl no.~259515. \fab{Mukherjee} was partially supported by the ISJRP project Mathematical Programming in Parameterized Algorithms. 
Brubach, Sankararaman, Srinivasan and Xu were supported by NSF Awards CCF-1422569, CCF-1749864, and CNS 1010789, a gift from Google, Inc., and by research awards from Adobe, Inc.\ and Amazon, Inc. A preliminary version of this paper appeared in \cite{AGM15} and \cite{BSSX16}.
}} \maketitle

\begin{abstract}
We consider the \emph{Stochastic Matching}
problem, which is motivated by applications in kidney exchange and
online dating. In this problem, we are given an undirected graph. Each edge
is assigned a known, independent probability of existence and a positive weight (or profit). We must \emph{probe} an edge to discover whether or not it exists. Each node is assigned a positive integer called a \emph{timeout} (or a \emph{patience}). On this random graph we are executing a process, which probes the edges one-by-one and gradually constructs a matching. The process is constrained in two ways. First, if a probed edge exists, it must be added irrevocably to the matching (the \emph{query-commit} model). Second, the timeout of a node $v$ upper-bounds the number of edges incident to $v$ that can be probed. The goal is to maximize the expected weight of the constructed matching.

For this problem, Bansal et al.~\cite{BGLMNR12algo} provided a $0.33$-approximation
algorithm for bipartite graphs and a $0.25$-approximation for general
graphs. We improve the approximation factors to $0.39$ 
and $0.269$, respectively. 
% $2.564$ and $3.709$

% We also consider an online version of the bipartite case, where one
% side of the partition arrives node by node, and each time a node $b$
% arrives we have to decide which edges incident to $b$ we want to
% probe, and in which order. Here we present a $0.245$-approximation,
% $4.07$-approximation,
% improving on the 
% $0.126$-approximation 
% $7.92$-approximation 
% of Bansal et al.~\cite{BGLMNR12algo}.

The main technical ingredient in our result is a novel way of probing
edges according to a \emph{not-uniformly-random} permutation. Patching
this method with an algorithm that works best for large-probability
edges (plus additional ideas) leads to our improved approximation
factors. 
\end{abstract}

\section{Introduction}
Maximum-weight matching is a fundamental problem in combinatorial optimization and has applications in a wide-range of areas such as market design \cite{roth2004kidney,abdulkadirouglu2005new}, computer vision \cite{caetano2009learning,van2011survey}, computational biology \cite{tian2006saga}, and machine learning \cite{sun2012efficient}. The basic version of this problem is well-understood; exact polynomial-time algorithms are available for both bipartite and general graphs due to the celebrated results of \cite{kuhn1955hungarian} and \cite{edmonds1965paths}, respectively. However, in many applications there are uncertainties associated with the input and typically, the problem of interest is more nuanced. A common approach to model this uncertainty is via randomness; we assume that we have a distribution over a collection of graphs. There are many such models ranging from stochastic edges \cite{costello2012stochastic,BDHPSS15} to stochastic vertices \cite{FMMM09,BSSX16}. 

In this paper, we study the stochastic-matching model where edges in the graph are uncertain and the exact realization of the graph is obtained by \emph{probing} the edges. 
%We work in the \emph{query-commit} model which forces the algorithm to include an edge in the matching if it is queried and realized. 
This model has many applications in kidney exchange, online dating, and labor markets (see Subsection~\ref{subsec:applications} for details). Further, we study a more general version of the problem, introduced by \cite{CIKMR09}, where the algorithm is constrained by the number of probes it can make on the edges incident to any single vertex. This is used to model the notion of \emph{timeouts} (also called \emph{patience}) which naturally arises in many applications.

	The formulation (described formally in subsection~\ref{subsec:probSt}) is similar to many other well-studied stochastic-optimization problems such as stochastic knapsack \cite{DGV08}, stochastic packing \cite{BGLMNR12algo,BSSX18},  and  stochastic shortest-path problems \cite{nikolova2006stochastic}.

%%%%%----- DEFINITIONS -----%%%%%

\subsection{Definitions and Notation}
	\label{subsec:probSt}
%	We now formally define the stochastic matching problem on bipartite graphs. In this problem, we are given a graph $G=(V, E)$, where $V$ denotes the set of vertices and $E$ denotes the set of potential edges.
	In the stochastic-matching problem, we are given a graph $G=(V, E)$, where $V$ denotes the set of vertices and $E$ denotes the set of potential edges. Additionally, we are given the following functions.
\begin{itemize}
	\item $p: E \rightarrow [0, 1]$ associates every edge $e$ with an independent probability of existence, $p_e$. When an edge is probed, it will exist with probability $p_e$ and must be added to the matching if it exists. Thus, we can only probe edges whose endpoints are currently unmatched.
	\item $w : E \rightarrow \mathbb{R}^+$ denotes a weight function that assigns a non-negative weight (or profit) $w_e$ to each edge $e$.
	\item $t: V \rightarrow \mathbb{N}$ is the \emph{timeout} (or \emph{patience}) function which sets an upper bound $t_v$ on the number of times a vertex $v$ can have one of its incident edges probed.
\end{itemize}

An algorithm for this problem probes edges in a possibly adaptive order. When an edge is probed, it is \emph{present} with probability $p_{e}$ (independent of all other edges), in which case it must be included in the matching under construction (\emph{query-commit} model) and provides a weight of $w_{e}$. We can probe at most $t_{v}$ edges among the set $\delta(v)$ of edges incident to a node $v$. Furthermore, when an edge $e$ is added to the matching, no edge $f\in\delta(e)$ (\ie incident on $e$) can be probed in subsequent steps. Each edge may only be probed once.  Our goal is to maximize the expected weight of the constructed matching.

	A naive approach to solve this problem is to construct an \emph{exponential-sized} Markov Decision Process (MDP) and solve it optimally using dynamic programming. However, there are no known algorithms to solve this exactly in \emph{polynomial time}. In fact, the exact complexity of computing the optimal solution is unknown. The naive solution above is in $\operatorname{PSPACE}$; it is unknown if the problem is in either $\operatorname{NP}$ or in $\operatorname{P}$.  Thus, following prior works \cite{CIKMR09,BGLMNR12algo,AGM15}, we aim at finding an \emph{approximation} to the optimal solution in polynomial time. To measure the performance of any algorithm, we use the standard notion of approximation ratio which is defined as follows.

	\begin{definition}[Approximation ratio]
	\label{def:approxRatio}
%		We use the notion of \emph{approximation ratio} to measure the performance of a randomized algorithm. 
		For any instance $\cI$, let $\mathbb{E}[\ALG(\cI)]$ denote the expected weight of the matching obtained by the algorithm $\ALG$ on $\cI$. Let $\mathbb{E}[\OPT(\cI)]$ denote the expected weight of the matching obtained by the optimal probing strategy. Then the approximation ratio of $\ALG$ is defined as $\min_{\cI} \frac{\mathbb{E}[\ALG(\cI)]}{\mathbb{E}[\OPT(\cI)]}$.
	\end{definition}

We remark that the measure used in \cite{AGM15} is the reciprocal of the ratio defined in Definition~\ref{def:approxRatio}. Thus, the ratios in \cite{AGM15} are always greater than $1$ while the ratios in this paper are at most $1$. Bansal et al.~\cite{BGLMNR12algo} provide an LP-based $0.33$-approximation when $G$ is bipartite, and via a reduction to the bipartite case, a $0.25$-approximation for general graphs (see also \cite{DBLP:conf/stacs/AdamczykSW14}).
	
	For the related stochastic-matching problem \emph{without} patience constraints (or equivalently, all $t_v$ equal infinity), the best-known algorithm achieves an approximation ratio of $1-1/\euler$ \cite{GKS19} and no algorithm can perform better than $0.898$ \cite{costello2012stochastic}. Since the problem without patience constraints is a special case of the problem studied here,  the latter hardness result applies to our setting as well.
	
%	When the patience constraints are not all equal to one or all infinite, it is believed that the best possible \emph{upper-bound} is $1/2$ \cite{CIKMR09}. 
	
	Chen~\etal \cite{CIKMR09} formulated and initiated the study of the problem with patience constraints and gave a probing scheme that achieves an approximation ratio of $0.25$ in \emph{unweighted} bipartite graphs. Later, Adamczyk~\cite{AdamczykuwGreedy} showed that the simple greedy algorithm (probe edges in non-increasing order of probability) achieves a 0.5-approximation for the unweighted case. This model was extended to \emph{weighted} bipartite graphs by \cite{BGLMNR12algo} who improved the ratio to $0.33$. \cite{AGM15} provided a new algorithm that further improved this ratio to $0.35$ which is the current state-of-the-art. For general graphs, the current best approximation is $0.31$~\cite{BCNSX15}.
	
	We note that our work and the most-recent prior work \cite{BGLMNR12algo,AGM15} use a natural linear program (LP) -- see (LP-BIP) in Section~\ref{sec:offline_bip} -- to upper bound the optimal solution. However, it was shown in~\cite{BGS19} that no algorithm can achieve an approximation better than $0.544$ using this LP even for the unweighted problem.

% In this paper we consider the \emph{Stochastic Matching} problem,
% which is motivated by applications in kidney exchange and online dating.
% Here we are given an undirected graph $G=(V,E)$. Each edge $e\in E$
% is labeled with an (existence) probability $p_{e}\in(0,1]$ and a weight
% (or profit) $w_{e}>0$, and each node $v\in V$ with a \emph{timeout}
% (or \emph{patience}) $t_{v}\in\mathbb{N}^{+}$. 
% An algorithm for this
% problem probes edges in a possibly adaptive order. Each time an edge
% is probed, it turns out to be \emph{present} with probability $p_{e}$,
% in which case it is (irrevocably) included in the matching under construction
% and provides a profit $w_{e}$. We can probe at most $t_{u}$ edges
% among the set $\delta(u)$ of edges incident to node $u$ (independently
% from whether those edges turn out to be present or absent). Furthermore,
% when an edge $e$ is added to the matching, no edge $f\in\delta(e)$
% (i.e., incident on $e$) can be probed in subsequent steps. Our goal
% is to maximize the expected weight of the constructed matching. Bansal
% et al.~\cite{BGLMNR12algo} provide an LP-based $3$-approximation
% when $G$ is bipartite, and via reduction to the bipartite case a
% $4$-approximation for general graphs (see also \cite{DBLP:conf/stacs/AdamczykSW14}).

\xhdr{Online Stochastic Matching with Timeouts}. We also consider the \emph{Online Stochastic Matching with Timeouts}
problem introduced in \cite{BGLMNR12algo}. Here we are given as input
a complete bipartite graph $G=(A\cup B,A\times B)$, where nodes in $B$ are
\emph{buyer types} and nodes in $A$ are \emph{items} that we wish
to sell. Like in the offline case, edges are labeled with probabilities
and profits, and nodes are assigned timeouts. However, in this case
timeouts on the item side are assumed to be unbounded. Then a second
bipartite graph is constructed in an online fashion. Initially this
graph consists of $A$ only. At each time step one random buyer $\tilde{b}$
of some type $b$ is sampled (possibly with repetition) from a given
probability distribution. The edges between $\tilde{b}$ and $A$
are copies of the corresponding edges in $G$. The online algorithm
has to choose at most $t_{b}$ unmatched neighbors of $\tilde{b}$,
and probe those edges in some order until some edge $a\tilde{b}$
turns out to be present (in which case $a\tilde{b}$ is added to the
matching and we gain the corresponding profit), or when all the mentioned
edges have been probed. This process is repeated $n$ times (with a new buyer being sampled at each iteration), and our goal
is to maximize the final total expected profit\footnote{As in \cite{BGLMNR12algo}, we assume that the probability of each buyer-type $b$ is an integer multiple of $1/n$.}.

For this problem, Bansal et al.~\cite{BGLMNR12algo} present a $0.126$-approximation
algorithm. In his Ph.D.~thesis, Li~\cite{li2011decision} claims an
improved $0.249$-approximation. 
However, his analysis contains a mistake \cite{Li14private}. By fixing that, he still achieves a $0.193$-approximation ratio improving over~\cite{BGLMNR12algo}. And although a corrected version of this result is not published later anywhere, in essence it would have follow the same lines as the result of Mukherjee~\cite{M19} who independently obtained exactly the same approximation ratio.

%In the appendix, we show that we can obtain an algorithm that achieves a ratio of $0.245$. However, since the publication of the conference version of this paper, this has been improved to $0.46$ (\cite{BSSX17}) and more recently to $0.51$ (\cite{fata2019multi}).
%We elaborate on the issue in Appendix~\ref{LiMistake}.

\subsection{Our Results}

Our main result is an approximation algorithm for bipartite Stochastic
Matching which improves the $0.33$-approximation of Bansal et al.~\cite{BGLMNR12algo}
(see Section \ref{sec:offline_bip}).

\begin{theorem} \label{thr:newalg}
There is an expected $0.39$-approximation algorithm for Stochastic
Matching in bipartite graphs. 
\end{theorem}
A $0.351$-approximation, originally presented in~ \cite{AGM15}, can be obtained as follows. We build upon the algorithm in \cite{BGLMNR12algo},
which works as follows. After solving a proper LP and rounding the
solution via a rounding technique from~\cite{GKPS06}, Bansal et
al.~probe edges in uniform random order. Then they show that every
edge $e$ is probed with probability at least $x_e\cdot g(p_{max})$, where $x_{e}$ is the fractional value of $e$ assigned by the LP, 
$p_{max}:=\max_{f\in\delta(e)}\{p_{f}\}$
is the largest probability of any edge incident to $e$ ($e$ excluded), and $g(\cdot)$ is a
decreasing function with $g(1)=1/3$.

Our idea is to instead consider edges in a carefully chosen \emph{non-uniform}
random order. This way, we are able to show (with a slightly simpler analysis) that each edge $e$ is probed with probability $x_{e}\cdot g\br{p_{e}}\geq \frac{1}{3}x_e$. Observe that we have the same function $g(\cdot)$ as in \cite{BGLMNR12algo}, but depending on $p_e$ rather than $p_{max}$. In particular, according to our analysis, small-probability edges are more likely to be probed than large-probability ones (for a given value of $x_{e}$), regardless of the
probabilities of edges incident to $e$. Though this approach alone
does not directly imply an improved approximation factor, we further patch it with a greedy algorithm that behaves best for large-probability edges, and this yields an improved approximation
ratio altogether. The greedy algorithm prioritizes large-probability edges for the same value of $x_e$.

We further improve the approximation factor for the bipartite case with the above mentioned updated patching algorithm, hence improving the ratio of $0.351$ from~\cite{AGM15} up to the current-best-known ratio of $0.39$ stated in Theorem~\ref{thr:newalg}.

We also improve on the $0.25$-approximation for general graphs in~\cite{BGLMNR12algo} (see Section \ref{sec:arbitrarygraphs}).
\begin{theorem} \label{thr:mainOfflineGeneral} 
There is an expected
$0.269$-approximation algorithm for Stochastic Matching in general
graphs. 
\end{theorem}
This is achieved by reducing the general case to the bipartite one
as in prior work, but we also use a refined LP with blossom inequalities
in order to fully exploit our large/small probability patching technique.

Similar arguments can also be successfully applied to the online case.
\begin{theorem} \label{thr:mainOnline} There is an expected $0.245$-approximation
algorithm for Online Stochastic Matching with Timeouts.
\end{theorem}

By applying our idea of non-uniform permutation of edges we would get
a $0.193$-approximation (the same as in~\cite{li2011decision}, after
correcting the mentioned mistake). However, due to the way edges have
to be probed in the online case, we are able to finely control the
probability that an edge is probed via \emph{dumping factors}. This
allows us to improve the approximation from $0.193$ to $0.24$. Our
idea is similar in spirit to the one used by Ma~\cite{DBLP:conf/soda/Ma14}
in his elegant $2$-approximation algorithm for correlated non-preemptive
stochastic knapsack. Further application of the large/small probability
trick gives an extra improvement up to $0.245$ (see Section \ref{sec:online}). We remark that since the publication of the conference version of this work, this has been improved to $0.46$ (\cite{BSSX17}) and more recently to $0.51$ (\cite{fata2019multi}).

\subsection{Applications}
	\label{subsec:applications}
	As previously mentioned, the stochastic matching problem is motivated from various applications such as kidney exchange and online dating. In this sub-section, we briefly consider these applications and show how we can use stochastic matching as a tool to solve these problems.
	
	\xhdr{Kidney exchange in the United States.} Kidney transplantation usually occurs from deceased donors; however, unlike other organs, another possibility is to obtain a kidney from a \emph{compatible} living donor since people need only one kidney to survive. There is a large waiting list of patients within the US who need a kidney for their survival. As of July 2019, the United Network for Organ Sharing (UNOS) estimates that the current number of patients who need a transplant is 113,265 with only 7,743 donors in the pool \footnote{Data obtained from \url{https://unos.org/data/transplant-trends/}}. One possibility is to enter the waitlist as a pair, where one person needs a kidney while the other person is willing to donate a kidney. Viewed as a stochastic-matching problem, the vertices are donor-patient pairs. An edge between two vertices $(u, v)$ exists, if the donor in $u$ is compatible with the patient in $v$ and the donor in $v$ is compatible with the patient in $u$. The probability on the edge is the probability that the exchange will take place. Before every transplant takes place, elaborate medical tests are usually performed which is very expensive. More specifically, as described in \cite{CIKMR09}, \emph{a test called the \emph{crossmatching} is performed, that combines the recipient's blood serum with some of donor's red blood cells and checks if the antibodies in the serum kill the cells.} This test is both expensive and time-consuming. Moreover, each exchange requires the transplant to happen \emph{simultaneously} since organ donation within the United States is \emph{at will}; donors are legally allowed to withdraw at any time including after agreeing for a donation. These constraints impose that for each donor-patient pair, the number of exchange initiations that can happen has to be small, which are modeled by the patience values at each vertex. Given the long wait-lists and the number of lives that depend on these exchanges, even small improvements to the accuracy of the algorithm have \emph{drastic} effects on the well-being of the population. Prior works (\eg \cite{BDHPSS15} and references therein) have empirically applied the variants of the stochastic-matching problem to real-world datasets; in fact, the current model that runs the US-wide kidney exchange is based on a stochastic-matching algorithm (\eg \cite{dickerson2017multi} and references therein).
	
	\xhdr{Online dating.} Online dating is quickly becoming the most popular form for couples to meet each other \cite{Dating19}. Platforms such as Tinder, eHarmony and Coffee meets Bagel generated about 1.7 billion USD in revenue in the year 2019\footnote{\url{https://www.statista.com/outlook/372/100/online-dating/worldwide}}. Suppose we have the case of a pool of heterosexual people, represented as the two vertex sets $U$ and $V$ in the bipartite graph. For every pair $u \in U$ and $v \in V$, the system learns their compatibility based on the questions they answer. The goal of the online-dating platform is to suggest couples that maximizes the social welfare (\ie the total number of matched couples). Each individual in a platform has a limited patience and thus, the system wants to ensure that the number of suggestions provided is small and limited. The stochastic-matching problem models this application where the probability on the edges represent the compatibility and the time-out function represents the individual patience. 
	
	\xhdr{Online labor markets.} In online labor markets such as Mechanical Turk, the goal is to match workers to tasks \cite{crowdsourcing14}. Each worker-task pair has a probability of completing the task based on the worker's skills and the complexity of the task. The goal of the platform is to match the pool of tasks to workers such that the (weighted) number of completed tasks is maximized. To keep workers in continued participation, the system needs to ensure that the worker is not matched with many tasks that they are incapable of handling. This once again fits in the model of the stochastic-matching problem where the workers and tasks represent the two sides of the bipartite graph, the edge-probability represents the probability that the task will be completed, and the time-out function represents the patience level of each worker.

%%%%%----- OTHER RELATED WORK -----%%%%%

\subsection{Other Related Work} Stochastic-matching problems come in many flavors and there is a long line of research for each of these models. The literature on the broader stochastic combinatorial optimization is (even more) vast (see \cite{swamy2006approximation} for a survey) and here we only mention some representative works.
	
	When the graph is unweighted and the time-out at every vertex is infinite, the classic $\operatorname{RANKING}$ algorithm of \cite{KVV90} gives an approximation ratio of $1-1/\euler$ for bipartite graphs; this in fact works even if the graph is unknown \emph{a priori}. For general graphs, the work of  \cite{costello2012stochastic} gives an algorithm that achieves a ratio of $0.573$; moreover, it shows that no algorithm can get a ratio better than $0.898$ for general graphs. The work of  \cite{molinaro2011query} gives an optimal algorithm in the special case of sparse graphs in this model. The paper \cite{GKS19} considers the \emph{weighted} version of this problem in bipartite graphs and designs algorithms that achieves an approximation ratio of $1-\frac{1}{\euler}$. (Recall that the timeouts are infinite in all of these works.) 
	
	The other line of research deals with the stochastic-matching problem where instead of a time-out constraint, the algorithm has to minimize the \emph{total} number of queried edges. The work of \cite{BDHPSS15} first proposed this model which was later considered and improved (by reducing the number of required queries) in many subsequent follow-up works including \cite{assadi2016stochastic,assadi2017stochastic,behnezhad2018almost,yamaguchi2018stochastic}.
	
	Online variants of the stochastic matching problems have been extensively studied due to their applications in Internet advertising and other Internet based markets. The paper \cite{FMMM09} introduced the problem of Online Matching with Stochastic Inputs. In this model, the vertices are drawn repeatedly and i.i.d.\ from a known distribution. The algorithm needs to find a match to a vertex each time one is presented,  \emph{immediately} and \emph{irrevocably}. The goal is to maximize the expected weight of the matching compared to an algorithm that knows the sequence of realizations \emph{a priori}. The work of \cite{FMMM09} gave an algorithm that achieves a ratio of $0.67$, which was subsequently improved by \cite{manshadi2012online,jaillet2013online,BSSX16}. Later work extended further to fully-online models of matching where both partitions of the vertex set in the bipartite graph are sampled i.i.d.\ from a known distribution~\cite{DSSX18,truong2019prophet}. The matching problem has also been studied in the \emph{two-stage} stochastic-optimization model \cite{katriel2007commitment}.
	
	The stochastic-matching problem is also related to the broader stochastic-packing literature, where the algorithm only knows a probability distribution over the item costs, and once it commits to include the item sees a realization of the actual costs \cite{DGV05,DGV08,BGLMNR12algo,BKNS12,BCNSX15,BSSX18}. Stochastic packing has also been studied in the online \cite{feldman2010online,devanur2019near,buchbinder2009design} and bandit \cite{guha2007approximation,badanidiyuru2018bandits,immorlica2018adversarial} settings.

\section{Stochastic Matching in Bipartite Graphs}
\label{sec:offline_bip}

In this section we present our improved approximation algorithm for Stochastic Matching in bipartite graphs. We start by presenting a simpler $0.351$-approximation in Section \ref{sub:Bipartite-graphs}, and then refine it in Section \ref{sec:newalg}.

\subsection{An Improved Approximation\label{sub:Bipartite-graphs}}

In this section we prove the following result.
\begin{theorem} \label{thr:mainOffline}
There is an expected $0.351$-approximation algorithm for Stochastic
Matching in bipartite graphs. 
\end{theorem}

Let $\OPT$ denote an optimal probing strategy and let $\ex{\OPT}$
denote its expected value. Consider the following LP: 
\begin{align}
\max & \sum_{e\in E}w_{e}p_{e}x_{e} & \br{\mbox{LP-BIP}}\label{lp:offlineBench}  \\
\mbox{s.t.} & \sum_{e\in\delta(v)}p_{e}x_{e}\leq1, & \forall v\in V;\label{cardinalityConstraints}\\
 & \sum_{e\in\delta(v)}x_{e}\leq t_{v}, & \forall v\in V;\label{toleranceConstraints}\\
 & 0\leq x_{e}\leq1, & \forall e\in E.
\end{align}
The proof of the following Lemma is already quite standard~\cite{DBLP:conf/stacs/AdamczykSW14,BGLMNR12algo,DGV08}
--- just note that $x_{e}=\pr{\OPT\mbox{ probes }e}$ is a feasible
solution of (LP-BIP). 
\begin{lemma}\label{lem:Bansal} \cite{BGLMNR12algo} 
Let $\LP_{bip}$
be the optimal value of \emph{(LP-BIP)}. It holds that $\LP_{bip}\geq\ex{\OPT}$.
\end{lemma}

Our approach is similar to the one of Bansal et al.~\cite{BGLMNR12algo}
(see also Algorithm \ref{alg:bipartite} in the figure). We solve
(LP-BIP): let $x=(x_e)_{e\in E}$ be the optimal fractional solution. Then we apply
to $x$ the rounding procedure by Gandhi et al.~\cite{GKPS06}, which
we shall call just GKPS. Let $\hat{E}$ be the set of rounded edges,
and let $\hat{x}_{e}=1$ if $e\in\hat{E}$ and $\hat{x}_{e}=0$ otherwise.
GKPS guarantees the following properties of the rounded solution: 
\begin{enumerate}
\item (Marginal distribution) For any $e\in E$, $\pr{\hat{x}_{e}=1}=x_{e}.$ 
\item (Degree preservation) For any $v\in V$, $\sum_{e\in\delta(v)}\hat{x}_{e}\leq\lceil\sum_{e\in\delta(v)}x_{e}\rceil\leq t_{v}.$ 
\item (Negative correlation) For any $v\in V$, any subset $S\subseteq\delta(v)$
of edges incident to $v$, and any $b\in\{0,1\}$, it holds that $\pr{\bigwedge_{e\in S}(\hat{x}_{e}=b)}\leq\prod_{e\in S}\pr{\hat{x}_{e}=b}.$ 
\end{enumerate}
Our algorithm sorts the edges in $\hat{E}$ according to a certain random
permutation and probes each edge $e\in\hat{E}$ according to this
order, but provided that the endpoints of $e$ are not matched already.
It is important to notice that, by the degree-preservation property, $\hat{E}$ has at most $t_{v}$ edges incident to each node
$v$. Hence, the timeout constraint of $v$ is respected even if the
algorithm probes all the edges in $\delta(u)\cap\hat{E}.$

Our algorithm differs from~\cite{BGLMNR12algo} and subsequent work
in the way edges are randomly ordered. Prior work exploits a 
uniformly-random order on $\hat{E}$. We rather use the following, more complex
strategy. For each $e\in \hat{E}$ we draw a random variable $Y_{e}$ distributed
on the interval $\brq{0,\frac{1}{p_{e}}\ln\frac{1}{1-p_{e}}}$ according
to the following cumulative distribution: $\pr{Y_{e}\leq y}=\frac{1}{p_{e}}\br{1-e^{-p_{e}y}}.$
Observe that %
\begin{comment}
always $Y_{e}\in\brq{0,\frac{1}{p_{e}}\ln\frac{1}{1-p_{e}}}$, and
that 
\end{comment}
the density function of $Y_{e}$ in this interval is $e^{-yp_{e}}$
(and zero otherwise). Edges of $\hat{E}$ are sorted in increasing order of the $Y_{e}$'s,
and they are probed according to that order. We 
let $Y$ denote the vector $(Y_{e})_{e\in \hat{E}}$, wherein the elements of $\hat{E}$ are ordered in some fixed manner. 

Define $\hat{\delta}(v):=\adj v\cap\E$. We say that an edge $e\in\hat{E}$
is \emph{safe} if, at the time we consider $e$ for probing, no other
edge $f\in\hat{\delta}(e)$ has already been taken into the matching. Note
that the algorithm can probe $e$ only in this case, and that if we do
probe $e$, it gets added to the matching with probability $p_{e}$  independent of all other events. 

\begin{algorithm}
\label{alg:bipartite} \caption{An algorithm for Bipartite Stochastic Matching $(\ALG_1)$.}
\begin{enumerate}
\item Let $\br{x_{e}}_{e\in E}$ be the solution to (LP-BIP). 
\item Round the solution $\br{x_{e}}_{e\in E}$ with GKPS; let $(\h x_{e})_{e\in E}$
be the rounded 0-1 solution, and $\h E=\{e\in E|\h x_{e}=1\}$. 
\item For every $e\in\h E$, sample a random variable $Y_{e}$ distributed
as $\pr{Y_{e}\leq y}=\frac{1-e^{-yp_{e}}}{p_{e}}$. 
\item For every $e\in\h E$ in increasing order of $Y_{e}$:

\begin{enumerate}
\item If no edge $f\in\hat{\delta}(e):=\delta(e)\cap\hat{E}$ is yet taken,
then probe edge $e$. 
\end{enumerate}
\end{enumerate}
\end{algorithm}

The main ingredient of our analysis is the following lower-bound on
the probability that an arbitrary edge $e$ is safe.\begin{lemma}
\label{lem:safe-bip} For every edge $e$ it holds that $\prcond{e\mbox{ is safe}}{e\in\hat{E}}\geq g\br{p_{e}}$,
where $$g\br p:=\frac{1}{2+p}\br{1-\exp\br{-\br{2+p}\frac{1}{p}\ln\frac{1}{1-p}}}.$$\end{lemma}

\begin{proof} In the worst case every edge $f\in\hat{\delta}(e)$
that is before $e$ in the ordering can be probed, and each of these
probes has to fail for $e$ to be safe. Thus 
\[
\prcond{e\mbox{ is safe}}{e\in\E}\geq\excondls{\E\setminus e,Y}{\prod_{f\in\hat{\delta}(e):Y_{f}<Y_{e}}\br{1-p_{f}}}{e\in\E}.
\]
Now we take expectation on $Y$ only, and using the fact that the
variables $Y_{f}$ are independent, we can write the latter expectation
as 
\begin{align}
 & \excondls{\E\setminus e}{\int_{0}^{\frac{1}{p_{e}}\ln\frac{1}{1-p_{e}}}\br{\prod_{f\in\hat{\delta}(e)}\br{\pr{Y_{f}\leq y}(1-p_{f})+\pr{Y_{f}>y}}}e^{-p_{e}\cdot y}\mbox{d}y}{e\in\E}.\label{eq:conditional}
\end{align}
Observe that $\pr{Y_{f}\leq y}\br{1-p_{f}}+\pr{Y_{f}>y}=1-p_{f}\pr{Y_{f}\leq y}.$
When $y>\frac{1}{p_{f}}\ln\frac{1}{1-p_{f}}$, then $\pr{Y_{f}\leq y}=1$,
and moreover, $\frac{1}{p_{f}}(1-e^{-p_{f}\cdot y})$ is an increasing
function of $y$. Thus we can upper-bound $\pr{Y_{f}\leq y}$ by $\frac{1}{p_{f}}(1-e^{-p_{f}\cdot y})$
for any $y\in\brq{0,\infty}$, and obtain that $1-p_{f}\pr{Y_{f}\leq y}\geq1-p_{f}\frac{1}{p_{f}}(1-e^{-p_{f}\cdot y})=e^{-p_{f}\cdot y}.$
Thus~\eqref{eq:conditional} can be lower-bounded by 
\begin{align*}
 & \excondls{\E\setminus e}{\int_{0}^{\frac{1}{p_{e}}\ln\frac{1}{1-p_{e}}}e^{-\sum_{f\in\hat{\delta}\br e}p_{f}\cdot y-p_{e}\cdot y}\mbox{d}y}{e\in\E}\\
= & \excondls{\E\setminus e}{\frac{1}{\sum_{f\in\hat{\delta}\br e}p_{f}+p_{e}}\br{1-e^{-\br{\sum_{f\in\hat{\delta}\br e}p_{f}+p_{e}}\frac{1}{p_{e}}\ln\frac{1}{1-p_{e}}}}}{e\in\E}.
\end{align*}

\end{proof}

We know from the negative-correlation and marginal-distribution properties that $\excondls{\E\setminus e}{\h x_{f}}{e\in\E}\leq\exls{\h E\setminus e}{\h x_{f}}=x_{f}$
for every $f\in\delta\br e$, and therefore $\excondls{\E\setminus e}{\sum_{f\in\hat{\delta}\br e}p_{f}}{e\in\h E}\leq\sum_{f\in\delta\br e}p_{f}x_{f}\leq2$,
where the last inequality follows from the LP constraints. Consider
function $f(x):=\frac{1}{x+p_{e}}\br{1-e^{-\br{x+p_{e}}\frac{1}{p_{e}}\ln\frac{1}{1-p_{e}}}}$.
This function is decreasing and convex. From Jensen's inequality we
know that $\ex{f(x)}\geq f(\ex{x})$. Thus 
\begin{multline*}
\excondls{\E\setminus e}{f\br{\sum_{f\in\hat{\delta}\br e}p_{f}}}{e\in\h E}\geq f\br{\excondls{\E\setminus e}{\sum_{f\in\hat{\delta}\br e}p_{f}}{e\in\h E}}\\
\hfill\geq f(2)=\frac{1}{2+p_{e}}\br{1-e^{-\br{2+p_{e}}\frac{1}{p_{e}}\ln\frac{1}{1-p_{e}}}}=g(p_{e}).\hfill\square
\end{multline*}

From Lemma \ref{lem:safe-bip} and the marginal distribution property,
the expected contribution of edge $e$ to the profit of the solution
is 
\[
w_{e}p_{e}\cdot\pr{e\in\hat{E}}\cdot\prcond{e\text{ is safe}}{e\in\hat{E}}\geq w_{e}p_{e}x_{e}\cdot g(p_{e})\geq w_{e}p_{e}x_{e}\cdot g(1)=\frac{1}{3}w_{e}p_{e}x_{e}.
\]
Therefore, our analysis implies a $1/3$ approximation, matching the
result in \cite{BGLMNR12algo}. However, by working with the probabilities
appropriately, we can do better as described next. 

\paragraph{Patching with Greedy.} We next describe an improved approximation algorithm, based on the
patching of the above algorithm with a simple greedy one. Let $\delta\in(0,1)$
be a parameter to be fixed later. We define $\lrg E$ as the set of (\emph{large})
edges $e$ with $p_{e}\geq\delta$, and let $E_{small}$ be the remaining
(\emph{small}) edges. Recall that $\LP_{bip}$ denotes the optimal
value of (LP-BIP). Let also $\lrg{\LP}$ and $\sml{\LP}$ be the fraction
of $\LP_{bip}$ due to large and small edges, respectively, i.e., $\lrg{\LP}=\sum_{e\in\lrg E}w_{e}p_{e}x_{e}$
and $\sml{LP}=\LP_{bip}-\lrg{\LP}$. Define $\gamma\in[0,1]$ such that
$\gamma \LP_{bip}=\LP_{large}$. By refining the above analysis, we
obtain the following result. \begin{lemma} \label{lem:bipRefined}
Algorithm \ref{alg:bipartite} has an expected approximation ratio $\frac{1}{3}\gamma+g(\delta)\br{1-\gamma}$.
\end{lemma}

\begin{proof} The expected profit of the algorithm is at least: 
\begin{multline*}
\sum_{e\in E}w_{e}p_{e}x_{e}\cdot g(p_{e})\geq\sum_{e\in E_{large}}w_{e}p_{e}x_{e}\cdot g(1)+\sum_{e\in E_{small}}w_{e}p_{e}x_{e}\cdot g(\delta)\\
=\frac{1}{3}LP_{large}+g(\delta)LP_{small}=\left(\frac{1}{3}\gamma+g(\delta)\br{1-\gamma}\right)\LP_{bip}.\hfill\square
\end{multline*}
%$\hfill\square$
\end{proof}

\xhdr{A greedy algorithm} ($\Gre$). Consider the following greedy algorithm $\Gre$. Compute a maximum weight
matching $M_{grd}$ in $G$ with respect to edge weights $w_{e}p_{e}$,
and probe the edges of $M_{grd}$ in any order. Note that the timeout
constraints are satisfied since we probe at most one edge incident
to each node (and timeouts are strictly positive by definition and
w.l.o.g.). 
\begin{lemma} \label{lem:greedy} $\Gre$
has an expected approximation ratio of at least $\delta\gamma$. \end{lemma}

\begin{proof} It is sufficient to show that the expected profit of
the obtained solution is at least $\delta\cdot\lrg{LP}$. Let $x=(x_{e})_{e\in E}$
be the optimal solution to (LP-BIP). Consider the solution $x'=(x'_{e})_{e\in E}$
that is obtained from $x$ by setting to zero all the variables corresponding
to edges in $E_{small}$, and by multiplying all the remaining variables
by $\delta$. Since $p_{e}\geq\delta$ for all $e\in\lrg E$, $x'$
is a feasible fractional solution to the following matching LP: 
\begin{align}
\max & \sum_{e\in E}w_{e}p_{e}z_{e} & \text{(LP-MATCH)}\label{lp:match} \\
\mbox{s.t.} & \sum_{e\in\delta(u)}z_{e}\leq1, & \forall u\in V;\nonumber \\
 & 0\leq z_{e}\leq1, & \forall e\in E.
\end{align}
The value of $x'$ in the above LP is $\delta\cdot \LP_{large}$ by
construction. Let $\LP_{match}$ be the optimal profit of (LP-MATCH).
Then $\LP_{match}\geq\delta\cdot \LP_{large}$. Given that the graph
is bipartite, (LP-MATCH) defines the matching polyhedron, and we can
find an integral optimal solution to it. But such a solution is exactly
a maximum weight matching according to weights $w_{e}p_{e}$, i.e. $\sum_{e\in M_{grd}}w_{e}p_{e}=\LP_{match}$. The claim follows since the expected profit of the greedy algorithm is precisely the weight of $M_{grd}$. $\hfill\square$ 
\end{proof}

\xhdr{A hybrid algorithm of $\ALG_1$ and $\Gre$}. The overall algorithm, denoted by $\hybrid (\delta)$ is stated as follows. For a given $\delta$, we simply compute the
value of $\gamma$, and run $\Gre$ if $\gamma\delta\geq\left(\frac{1}{3}\gamma+g(\delta)\br{1-\gamma}\right)$,
and $\ALG_1$ otherwise\footnote{Note that we cannot run both algorithms and take the better solution, due to the probe-commit constraint}.

The approximation factor of $\hybrid(\delta)$ is given by $\max\{\frac{\gamma}{3}+(1-\gamma)g(\delta),\gamma\delta\}$,
and the worst case is achieved when the two quantities are equal,
i.e., for $\gamma=\frac{g\br{\delta}}{\delta+g\br{\delta}-\frac{1}{3}}$,
yielding an approximation ratio of $\frac{\delta\cdot g\br{\delta}}{\delta+g\br{\delta}-\frac{1}{3}}$.
\iffalse
Maximizing (numerically) the latter function in $\delta$ gives $\delta=0.6022$,
and the final $\fab{0.351}$-approximation ratio claimed in Theorem~\ref{thr:mainOffline}.
\fi
Now we just need to maximize this ratio. Since this function is quite complicated and so finding algebraically its maximum seems impossible, we need to compute it numerically. To avoid issues of numerical error, let us just notice that for $\delta=0.6$, the ratio is approximately $0.351563$ -- which allows us to claim the ratio of $0.351$ from Theorem~\ref{thr:mainOffline}.

%%%%%----- NEW BIPARTITE ALGORITHM -----%%%%%

\subsection{A Refined Approximation}
	\label{sec:newalg}
%	{\color{red} This section is reorganized by Pan. Take a look!}

	We now describe the approach to achieve the $0.39$-approximation ratio stated in Theorem~\ref{thr:newalg}. The main algorithm, denoted by $\MAIN$, consists of two sub-routines. One is $\ALG_1$ as described in Algorithm~\ref{alg:bipartite} in Section~\ref{sub:Bipartite-graphs}. The other is a new \emph{patching} algorithm, denoted by $\ALG_2$, which is described in Algorithm~\ref{alg:ALG2}.
			\begin{algorithm}[!h]
			\caption{Patching algorithm ($\ALG_2$)}
				\label{alg:ALG2}
			\DontPrintSemicolon
		Construct and solve (LP-BIP) for the input instance. Let $\vec{x}$ be an optimal solution. \;
		Consider the vector $\vec{y}$ such that $y_e := p_e x_e$ for every edge $e \in E$. \;
		Run $\GKPS$ with $\vec{y}$ as the input to obtain an integral vector $\vec{Y}$. \;
		Probe all edges $e \in E$ with $Y_e=1$. \;
	\end{algorithm}

Let $\vec{x}$ be the optimal solutions to (LP-BIP). Recall the definition of function $g: [0, 1] \rightarrow [0, 1]$ from Lemma~\ref{lem:safe-bip}:
\[
g(p) := \frac{1}{2+p} \left( 1- \exp\left( - \frac{2+p}{p} \ln \frac{1}{1-p} \right) \right).
\]
From the same Lemma~\ref{lem:safe-bip} we know that the expected total weight achieved by $\ALG_1$ is 
\[
\ex{ALG_1} \geq \sum_{e \in E} w_e p_e x_e\cdot g(p_e).
\]
Consider now vector $\vec{y}$ defined as $y_e := p_ex_e$. Notice that vector $\vec{y}$ is a feasible solution to (LP-MATCH). And in  Algorithm~\ref{alg:ALG2} we can see that it is $\vec{y}$ that guides $\ALG_2$.

We shall prove in a moment that the expected outcome of $\ALG_2$ is
\[
\ex{ALG_2} = \sum_{e \in E} w_e p^2_e x_e.
\]

Our main algorithm, $\MAIN$, is formally stated in Algorithm~\ref{alg:ALG3}.
			\begin{algorithm}[!h]
			\caption{The main algorithm ($\MAIN$).}
				\label{alg:ALG3}
			\DontPrintSemicolon
		Construct and solve (LP-BIP) for the input instance; let $\vec{x}$ be its optimal solution. \;
	Run $\ALG_1$ if $\sum_{e \in E} w_e p_e x_e\cdot g(p_e) \geq \sum_{e \in E} w_e p^2_e x_e$; otherwise run $\ALG_2$.
	\end{algorithm}

%\xhdr{The main algorithm} ($\MAIN$). 
	
%	The main algorithm will use one of these two algorithms based on a specific criteria which will be described later in the section. \bri{Should note where this is described.} $\ALG_1$ is the \textbf{Algorithm 1} from \cite{AGM15}. For completeness we describe this algorithm in Algorithm~\ref{alg:AGMAlg1}. Our key contribution is the new \emph{patching} algorithm $\ALG_2$ which is described in Algorithm~\ref{alg:ALG2}.

%		\begin{algorithm}[!h]
%			\caption{$\ALG_1$ (\cite{AGM15})}
%				\label{alg:AGMAlg1}
%			\DontPrintSemicolon
%		Construct and solve the benchmark LP~\eqref{lp:offlineBench} for the input instance. \;
%		Use Dependent Rounding to round this to obtain a random integral solution $\vec{X}$. Let $\hat{E} :=\{e: X_e = 1\}$. \;
%		For every edge $e \in \hat{E}$, sample a random variable $Y_e$ such that $\Pr[Y_e \leq y] = \frac{1-e^{-y p_e}}{p_e}$. \;
%		Consider each edge $e \in \hat{E}$ in increasing order of $Y_e$. If no edge in $\Neigh{e} \cap \hat{E}$ has been added to the matching, probe $e$. \;
%	\end{algorithm}

To prove Theorem~\ref{thr:newalg}, it remains to bound the expected outcome of $\ALG_2$.
 We then use lowerbounds on $\ex{ALG_1}$ and $\ex{ALG_2}$ to bound the approximation achieved by $\MAIN$. 		
%		Let $\Elarge(\delta)$ denote the set of edges $e \in E$ such that $p_e \geq \delta$. Let $\LPOPT$ denote the optimal value of $\LP$~\eqref{lp:offlineBench}. Let $\LPlarge$ denote the contribution by the edges in $\Elarge(\delta)$ (\ie $\LPlarge := \sum_{e \in \Elarge} w_e p_e x_e$) and let $\LPsmall := \LPOPT - \LPlarge$. Define $\gamma \in [0, 1]$ such that $\gamma := \frac{\LPOPT}{\LPlarge}$. 
We now lower bound the profit of $\ALG_2$ in the following lemma.
		\begin{lemma}
			\label{lem:ALG2}
			The total expected weight of the matching obtained by $\ALG_2$ is $\sum_{e \in E} w_e p^2_e x_e$ where $\vec{x}$ denotes the optimal solution to (LP-BIP).
		\end{lemma}
		\begin{proof}
Recall that $\vec{y} := \vec{p} \cdot \vec{x}$ is a feasible solution to (LP-MATCH): so, the use of the GKPS dependent-rounding procedure on the polytope from (LP-MATCH) is allowed. First, from property (P1) of $\GKPS$, the probability that an edge $e \in E$ has $Y_e = 1$ is $y_e = p_e x_e$. Second, from property (P2) and the fact that $\sum_{e \in \Neigh{v}} y_e = \sum_{e \in \Neigh{v}} p_e x_e \leq 1$, we have that the subgraph induced by the edges with $Y_e = 1$ has at most one edge incident to any vertex $v \in V$. Thus, given that an edge $e$ has $Y_e = 1$, it is guaranteed to be probed and the probability that it is eventually chosen into the matching is its probability of existing $p_e$. Putting these two facts together, the probability that any edge $e \in E$ is included in the final matching is $p^2_e x_e$. Using the linearity of expectation, we obtain that the total expected weight of the matching is $\sum_{e \in E} w_e p^2_e x_e$.
		\end{proof}		
	%\xhdr{Proof of Theorem~\ref{thr:newalg}.}
We now have all the ingredients to prove Theorem~\ref{thr:newalg}.
	
	\begin{proof}[of Theorem~\ref{thr:newalg}] Algorithm $\MAIN$ chooses either $\ALG_1$ or $\ALG_2$ depending on the maximum of the lowerbounds  of $\ALG_1$ and $\ALG_2$. Hence to find its worst case behaviour we have to characterize an instance that minimizes 
\[
\max \left\{\sum_{e \in E} w_e p_e x_e \cdot g(p_e), \sum_{e \in E} w_e p^2_e x_e \right\}.
\] 
Let $E(q)$ denote the set of edges such that $p_e = q$. Let $h_q := \sum_{e \in E(q)} w_e x_e$.  Notice that the quantity $\int_{q=0}^1   q h_q \mbox{d}q$ is just the value of (LP-BIP), which is our upperbound on $OPT$.
The outcome of $\ALG_1$ is thus $\int_{q=0}^1 g(q) \cdot  q h_q \mbox{d}q$. And the value of (LP-MATCH) induced by $\vec{y}$ is $\int_{q=0}^1   q^2 h_q \mbox{d}q$, which is at the same time the outcome of $\ALG_2$.
Thus, the adversary wants to minimize the following mathematical program.
		
		\begin{equation}
				\label{lp:adversary}
				\begin{array}{ll@{}ll}
					\text{minimize} \qquad & \displaystyle \max\left \{ \int_{q=0}^1 g(q) \cdot  q h_q \mbox{d}q, \int_{q=0}^1 q^2 h_q \mbox{ d}q \right \} & \quad \text{such that}\\
					& \displaystyle  \int_{q=0}^1 q h_q \mbox{ d}q = 1 & 
			\end{array}
		\end{equation}
		
		The normalization constraint $\int_{q=0}^1 q h_q \mbox{ d}q = 1$ ensures that the optimal value to the adversarial program \eqref{lp:adversary} is the approximation ratio for the algorithm $\MAIN$. 

Such a mathematical program may be hard to solve in full generality, due to the fact that the variable over which we optimize is in fact a (not necessarily continuous) probability distribution. However, in our case we can characterize the optimal solution, i.e., function $q h_q$ such that $\int_{q=0}^1 q h_q \mbox{ d}q = 1$, algebraically.
Let us next formulate our problem in a more compact way where we define $f := qh_q$:
		\begin{equation}\label{lp:adversary2}
					\min_f  \left \{ \max\br{ \int_{q=0}^1 g(q) \cdot  f(q) \mbox{d}q, \int_{q=0}^1 q \cdot f(q) \mbox{d}q }  \ \text{s.t.} 					 \int_{q=0}^1 f(q) \mbox{d}q = 1 \right \}.
		\end{equation}

Since $g$ is a concave function, it is point-wise at least as large as a linear function $\br{1-q}g(0) + q\cdot g\br{1}$ for all $q \in \br{0,1}$. Hence, the minimum of~\eqref{lp:adversary2} is at least as large as the minimum of the following program: 
		\begin{align}
					\min_f   &\quad \max\br{ \int_{q=0}^1 \br{\br{1-q}g(0) + q\cdot g\br{1}} \cdot  f(q) \mbox{d}q, \int_{q=0}^1 q \cdot f(q) \mbox{d}q }  \label{lp:adversary3}\\
 \text{s.t.} 	&\quad				 \int_{q=0}^1 f(q) \mbox{d}q = 1 .\nonumber
		\end{align}
Here we have two linear functions, i.e., $\br{1-q}g(0) + q\cdot g\br{1}$ and $q$. This allows us to simplify it further:
\begin{align*}
\int_{q=0}^1 \br{\br{1-q}g(0) + q\cdot g\br{1}} \cdot  f(q) \mbox{d}q 
=  g\br{0} + \br{g\br{1}-g\br{0}} \int_{q=0}^1 q\cdot f\br{q} \mbox{d}q.
\end{align*}

\iffalse
\begin{figure}
\includegraphics[width=4.5in]{Figure_1.pdf}
\caption{\ }
\label{fig:threefunctions}
\end{figure}
\fi

Even though the variable $f$ in program~\eqref{lp:adversary3} is a density function, we can consider the whole integral $\int_{q=0}^1 q\cdot f\br{q} \mbox{d}q$ as a single real variable from $\brq{0,1}$, and the program simplifies to
		\begin{align*}
					\min_{\alpha}   &\quad \max\br{ g\br{0} + \br{g\br{1}- g\br{0}}\cdot \alpha, \alpha }  \\
 \text{s.t.} 	&\quad				\alpha \in \brq{0,1} .\nonumber
		\end{align*}

Since function $\alpha$ is increasing and function $g\br{0} + \br{g\br{1}- g\br{0}}\cdot \alpha$ is decreasing, the minimum is obtained for $\alpha$ for which $g\br{0} + \br{g\br{1}- g\br{0}}\cdot \alpha = \alpha$. This yields a value of $\alpha$ such that 
\[
\alpha = \frac{g\br{0}}{1+g\br{0}-g\br{1}} = \frac{\frac{1}{2}\br{1-e^{-2}}}{1+\frac{1}{2}\br{1-e^{-2}}-\frac{1}{3}} = \frac{3e^2-3}{7e^2-3} \approx 0.39338739.
\]
Solution $\alpha$ which is a number has to be translated into solution $f$ of program~\eqref{lp:adversary3} which is a density function. This is however straightforward: $f$ is a density function which places $\alpha$ mass on point 1, and $1-\alpha$ mass on point 0. At the same time $f$ is a solution to the initial program~\eqref{lp:adversary2}. And since the value of program~\eqref{lp:adversary2} for such an $f$ is also equal $\frac{3e^2-3}{7e^2-3}$ we conclude that it is the actual minimal value of it.

\iffalse
Discretizing the integrals using a $\epsilon$-mesh on the $[0, 1]$ interval and using a numerical solver, we obtain that the optimal value is at least $0.39$.\fabr{I think we need to give more details here. Again, a mistake could mean a better result} Moreover, the worst case instance has the following structure: the adversary has two types of edges namely \emph{small} and \emph{large edges}. The small edges all have the same probability value $p$ which is arbitrarily close to $0$ and the large edges all have a probability of $1$. 
\fi
		\end{proof}

%%%%%----- END NEW ALGORITHM STUFF -----%%%%%

\section{Stochastic Matching in General Graphs}
\label{sec:arbitrarygraphs}

In this section, we present our improved approximation algorithm for Stochastic Matching in general graphs as stated in Theorem~\ref{thr:mainOfflineGeneral}.

%\subsection{General graphs (Proof of Theorem~\ref{thr:mainOfflineGeneral})}
%\label{sec:arbitrarygraphs}

%For general graphs, 

We consider the linear program LP-GEN which is
obtained from LP-BIP by adding the following \emph{blossom inequalities}:
\begin{align}
 & \sum_{e\in E(W)}p_{e}x_{e}\leq\frac{|W|-1}{2} & \forall W\subseteq V,|W|\text{ odd}.\label{blossomConstraints}
\end{align}
Here $E(W)$ is the subset of edges with both endpoints in $W$. We remark that, using standard tools from matching theory, we can solve LP-GEN in polynomial time despite its exponential number of constraints~\cite{Schrijver:book}. Also, in this case, $x_{e}=\pr{\OPT\mbox{ probes }e}$ is a feasible
solution of LP-GEN, hence the analogue of Lemma \ref{lem:Bansal} still holds.

Our stochastic-matching algorithm for the case of a general graph $G=\br{V,E}$
works via a reduction to the bipartite case. First we solve LP-GEN;
let $x=\br{x_{e}}_{e\in E}$ be the optimal fractional solution. Second we randomly split
the nodes $V$ into two sets $A$ and $B$, with $E_{AB}$ being
the set of edges between them. On the bipartite graph $\br{A\cup B,E_{AB}}$
we apply the algorithm for the bipartite case, but using the fractional
solution $\br{x_{e}}_{e\in E_{AB}}$ induced by LP-GEN rather than
solving LP-BIP. Note that $\br{x_{e}}_{e\in E_{AB}}$ is a feasible
solution to LP-BIP for the bipartite graph $\br{A\cup B,E_{AB}}$.

The analysis differs only in two points w.r.t.~the one for the bipartite
case. First, with $\h E_{AB}$ being the subset of edges of $E_{AB}$
that were rounded to 1, we have now that $\pr{e\in\h E_{AB}}=\pr{e\in E_{AB}}\cdot\prcond{e\in\h E_{AB}}{e\in E_{AB}}=\frac{1}{2}x_{e}$.
Second, but for the same reason, using again the negative correlation and marginal distribution properties, we have
\begin{align*}
 & \excondls{}{\sum_{f\in\h{\delta}\br e}p_{f}}{e\in\h E_{AB}}\leq \sum_{f\in\delta\br e}p_{f}\pr{f\in\h E_{AB}}=\sum_{f\in\delta\br e}\frac{p_{f}x_{f}}{2}\leq\frac{2-2p_{e}x_{e}}{2}\leq1.
\end{align*}
Repeating the steps of the proof of Lemma~\ref{lem:safe-bip} and
including the above inequality we get the following. \begin{lemma}
\label{lem:safe-gen} For every edge $e$ it holds that $\prcond{e\mbox{ is safe}}{e\in\hat{E}_{AB}}\geq h\br{p_{e}}$,
where $$h\br{p}:=\frac{1}{1+p}\br{1-\exp\br{-\br{1+p}\frac{1}{p}\ln\frac{1}{1-p}}}.$$
\end{lemma}

Since $h(p_{e})\geq h(1)=\frac{1}{2}$, we directly obtain a $1/4$-approximation
which matches the result in \cite{BGLMNR12algo}. Similarly to the
bipartite case, we can patch this result with the simple greedy algorithm
(which is exactly the same in the general graph case). For a given
parameter $\delta\in[0,1]$, let us define $\gamma$ analogously to
the bipartite case. Similarly to the proof of Lemma \ref{lem:bipRefined},
one obtains that the above algorithm has approximation factor $\frac{\gamma}{4}+\frac{1-\gamma}{2}h(\delta)$.
Similarly to the proof of Lemma \ref{lem:greedy}, the greedy algorithm
has approximation ratio $\gamma\delta$ (here we exploit the blossom
inequalities that guarantee the integrality of the matching polyhedron).
We can conclude similarly that in the worst case $\gamma=\frac{h\br{\delta}}{2\delta+h\br{\delta}-1/2}$,
yielding an approximation ratio of $\frac{\delta\cdot h\br{\delta}}{2\delta+h\br{\delta}-1/2}$.
Maximizing (numerically) this function over $\delta$ gives, for $\delta=0.5580$,
the $0.269$ approximation ratio claimed in Theorem \ref{thr:mainOfflineGeneral}.

%\newpage

%\appendix

\section{Online Stochastic Matching with Timeouts}
\label{sec:online}

Let $G=\br{A\cup B,A\times B}$ be the input graph, with items $A$
and buyer types $B$. We use the same notation for edge probabilities,
edge profits, and timeouts as in Stochastic Matching. %\mar{In
%what follows, by `proposing buyer $b$ item $a$' we mean `probing
%edge $ab$'}\marr{maybe it should also be stressed in the intro}.
Following \cite{BGLMNR12algo}, we can assume w.l.o.g. that each buyer
type is sampled uniformly with probability $1/n$. Consider the following
linear program:

\begin{align*}
\max & \sum_{a\in A,b\in B}w_{ab}p_{ab}x_{ab} & \mbox{(LP-ONL)}\\
\mbox{s.t.} & \sum_{b\in B}p_{ab}x_{ab}\leq1, & \forall a\in A\\
 & \sum_{a\in A}p_{ab}x_{ab}\leq1, & \forall b\in B\\
 & \sum_{a\in A}x_{ab}\leq t_{b}, & \forall b\in B\\
 & 0\leq x_{ab}\leq1, & \forall ab\in E.
\end{align*}
%\begin{eqnarray*}
%\mbox{LP-ONL: }\max &  & \sum_{a\in A,b\in B}w_{ab}p_{ab}x_{ab}\\
%\mbox{s.t.} &  & \forall_{a\in A}\sum_{b\in B}p_{ab}x_{ab}\leq1,\quad\forall_{b\in B}\sum_{a\in A}p_{ab}x_{ab}\leq1,\quad\forall_{b\in B}\sum_{a\in A}x_{ab}\leq t_{b}.
%\end{eqnarray*}
The above LP models a bipartite stochastic-matching instance
where one side of the bipartition contains exactly one buyer per buyer
type. In contrast, in the online case, several buyers of the same buyer
type (or none at all) can arrive, and the optimal strategy can allow
many buyers of the same type to probe edges. %Note that in the above LP we have one node per buyer type, while in
%the algorithm many buyers of one type can come, and the optimal strategy
%can allow many buyers of the same type to probe items.
This is not a problem though, since the following lemma from \cite{BGLMNR12algo}
allows us just to look at the graph of buyer types and not at the
actual realized buyers. \begin{lemma} \label{lem:onlineBansal} (\cite{BGLMNR12algo},
Lemmas 9 and 11) Let $\ex{\OPT}$ be the expected profit of the optimal
online algorithm for the problem. Let $LP_{onl}$ be the optimal value
of \emph{LP-ONL. }It holds that $\ex{\OPT}\leq LP_{onl}$. \end{lemma}

\global\long\def\Ab{A_{b}}

We will devise an algorithm whose expected outcome is at least $0.245 \cdot LP_{onl}$,
and then Theorem \ref{thr:mainOnline} follows from Lemma \ref{lem:onlineBansal}.

\paragraph{The algorithm.}

We initially solve LP-ONL and let $\br{x_{ab}}_{ab\in A\times B}$
be the optimal fractional solution. Then buyers arrive. When a buyer
of type $b$ is sampled, then: (a) if a buyer of the same type $b$
was already sampled before we simply discard her, do nothing, and
wait for another buyer to arrive, and (b) if it is the first buyer of type
$b$, then we execute the following \emph{subroutine for buyers.}
Since we take action only when the first buyer of type $b$ comes,
we shall denote such a buyer simply by $b$, as it will not cause
any confusion.

\paragraph{Subroutine for buyers.}

Let us consider the step of the online algorithm in which the first
buyer of type $b$ arrived, if any. Let $\Ab$ be the items that are still
available when $b$ arrives. Our subroutine will probe a subset of
at most $t_{b}$ edges $ab$, $a\in\Ab$. Consider the vector $\br{x_{ab}}_{a\in\Ab}$.
Observe that it satisfies the constraints $\sum_{a\in\Ab}p_{ab}x_{ab}\leq1$
and $\sum_{a\in\Ab}x_{ab}\leq t_{b}$. Again using GKPS, we round
this vector in order to get $\br{\hat{x}_{ab}}_{a\in\Ab}$ with $\hat{x}_{ab}\in\{0,1\}$,
and satisfying the marginal distribution, degree preservation, and
negative correlation properties\footnote{In this case, we have a bipartite graph where one side has
only one vertex, and here GKPS reduces to Srinivasan's
rounding procedure for level-sets \cite{DBLP:conf/focs/Srinivasan01}.}. %$\sum_{a\in\Ab}\hat{x}_{ab}\leq t_{b}$.
Let $\hat{A}_{b}$ be the set of items $a$ such that $\hat{x}_{ab}=1$.
For each $ab$, $a\in\h A_{b}$, we independently draw a random variable
$Y_{ab}$ with distribution: $\pr{Y_{ab}<y}=\frac{1}{p_{ab}}\br{1-\exp\br{-p_{ab}\cdot y}}$
for $y\in\brq{0,\frac{1}{p_{ab}}\ln\frac{1}{1-p_{ab}}}$. Let $Y=\br{Y_{ab}}_{a\in\h A_{b}}$.

Next we consider items of $\hat{A}_{b}$ in increasing order of $Y_{ab}$.
Let $\alpha_{ab}\in[\frac{1}{2},1]$ be a \emph{dumping factor} that
we will define later. With probability $\alpha_{ab}$ we probe edge
$ab$ and as usual we stop the process (of probing edges incident
to $b$) if $ab$ is present. Otherwise (with probability $1-\alpha_{ab}$)
we \emph{simulate} the probe of $ab$, meaning that with probability
$p_{ab}$ we stop the process anyway --- like if edge $ab$ were probed
and turned out to be present. Note that we do not get any profit from
the latter simulation since we do not really probe $ab$.

\paragraph{Dumping factors.}

It remains to define the dumping factors. For a given edge $ab$,
let 
\[
\beta_{ab}:=\excondls{\h A{}_{b}\setminus a,Y}{\prod_{a'\in A_{b}:Y_{a'b}<Y_{ab}}\br{1-p_{a'b}}}{a\in\h A_{b}}.
\]
%It is so, because in this simple case of taking at most one item,
%all the other edges of $\hat{A}_{b}$ before $ab$ in the random order
%are indeed probed. 
Using the inequality $\sum_{a\in\Ab}p_{ab}x_{ab}\leq1$, by repeating the analysis from Section~\fab{\ref{sec:offline_bip}} we can
show that 
\[
\beta_{ab}\geq h(p_{ab})=\frac{1}{1+p_{ab}}\br{1-\exp\br{-\br{1+p_{ab}}\frac{1}{p_{ab}}\ln\frac{1}{1-p_{ab}}}}\geq\frac{1}{2}.
\]
Let us assume for the sake of simplicity that we are able to compute
$\beta_{ab}$ exactly. 
%We will show in Appendix \ref{sec:computeDumping} how to remove
%this assumption. 
We set $\alpha_{ab}=\frac{1}{2\beta_{ab}}$. Note
that $\alpha_{ab}$ is well defined since $\beta_{ab}\in[1/2,1]$.

\paragraph{Analysis.}

Let us denote by ${\cal A}_{b}$ the event that at least one buyer
of type $b$ arrives. The probability that an edge $ab$ is probed
can be expressed as: 
\[
\pr{{\cal A}_{b}}\cdot\prcond{\mbox{no }b'\mbox{ takes }a\mbox{ before }b}{{\cal A}_{b}}\cdot\prcond{b\mbox{ probes }a}{{\cal A}_{b}\wedge a\mbox{ is not yet taken}}.
\]
The probability that $b$ arrives is $\pr{{\cal A}_{b}}=1-\br{1-\frac{1}{n}}^{n}\geq1-\frac{1}{e}$.
We shall show first that $$\prcond{b\mbox{ probes }a}{{\cal A}_{b}\wedge a\mbox{ is not yet taken}}$$
is exactly $\frac{1}{2}x_{ab}$, and later we shall show that $\prcond{\mbox{no }b'\mbox{ takes }a\mbox{ before }b}{{\cal A}_{b}}$
is at least $\frac{1}{1+\frac{1}{2}\br{1-\frac{1}{e}}}$. This will
yield that the probability that $ab$ is probed is at least 
\[
\br{1-\frac{1}{e}}\frac{1}{1+\frac{1}{2}\br{1-\frac{1}{e}}}\cdot\frac{1}{2}x_{ab}=\frac{e-1}{3e-1}x_{ab}>0.24 x_{ab}.
\]

Consider the probability that some edge $a'b$ appearing before $ab$
in the random order \emph{blocks} edge $ab$, meaning that $ab$ is
not probed because of $a'b$. Observe that each such $a'b$ is indeed
considered for probing in the online model, and the probability that
$a'b$ blocks $ab$ is therefore $\alpha_{a'b}p_{a'b}+(1-\alpha_{a'b})p_{a'b}=p_{a'b}$.
We can conclude that the probability that $ab$ is not blocked is
exactly $\beta_{ab}$.

Due to the dumping factor $\alpha_{ab}$, the probability that we
actually probe edge $ab\in\hat{A}_{b}$ is exactly $\alpha_{ab}\cdot\beta_{ab}=\frac{1}{2}$.
Recall that $\pr{a\in\hat{A}_{b}}=x_{ab}$ by the marginal distribution
property. Altogether 
\begin{equation}
\prcond{b\mbox{ probes }a}{{\cal A}_{b}\wedge a\mbox{ is not yet taken}}=\frac{1}{2}x_{ab}.\label{eq:afterdumping}
\end{equation}

Next let us condition on the event that buyer $b$ arrived and lower-bound the probability that $ab$ is not blocked on the $a$'s side in
such a step, i.e., that no other buyer has taken $a$ already. The
buyers, who are first occurrences of their type, arrive uniformly
at random. Therefore, we can analyze the process of their arrivals
as if it was constructed by the following procedure: every buyer $b'$
is given an independent random variable $Y_{b'}$ distributed exponentially
on $[0,\infty]$, i.e., $\pr{Y_{b'}<y}=1-e^{y}$; buyers arrive in
increasing order of their variables $Y_{b'}$. Once buyer $b'$ arrives,
it probes edge $ab'$ with probability (exactly) $\alpha_{ab'}\beta_{ab'}x_{ab'}=\frac{1}{2}x_{ab'}$
--- these probabilities are independent among different buyers. Thus,
conditioning on the fact that $b$ arrives, we obtain the following
expression for the probability that $a$ is safe at the moment when
$b$ arrives: 
\begin{eqnarray*}
 &  & \prcond{\mbox{no }b'\mbox{ takes }a\mbox{ before }b}{{\cal A}_{b}}\\
 & \geq & \excondls{}{\prod_{b'\in B\setminus b:Y_{b'}<Y_{b}}\br{1-\prcond{{\cal A}_{b'}}{{\cal A}_{b}}\prcond{b'\mbox{ probes }ab'}{{\cal A}_{b'}}p_{ab'}}}{{\cal A}_{b}}\\
 & = & \int_{0}^{\infty}\prod_{b'\in B\setminus b}\br{1-\prcond{{\cal A}_{b'}}{{\cal A}_{b}}\cdot\prcond{Y_{b'}<y}{{\cal A}_{b'}}\cdot\prcond{b'\mbox{ probes }ab'}{{\cal A}_{b'}}p_{ab'}}e^{-y}\mbox{d}y.
\end{eqnarray*}
Now let us upper-bound each of the probability factors in the above
product. First of all $\prcond{{\cal A}_{b'}}{{\cal A}_{b}}=1-\br{1-\frac{1}{n}}^{n-1}\leq1-\frac{1}{e}$.
Second, $\prcond{Y_{b'}<y}{{\cal A}_{b'}}=1-e^{-y}$ just by definition\footnote{The ${\cal A}_{b'}$ event in the condition simply indicates that
$Y_{b'}$ was drawn.}. Third, from~\eqref{eq:afterdumping} we have that $\prcond{b'\mbox{ probes }ab'}{{\cal A}_{b'}}=\frac{x_{ab}}{2}.$

Thus the above integral can be lower-bounded by 
\begin{eqnarray*}
 &  & \int_{0}^{\infty}\prod_{b'\in B\setminus b}\br{1-\br{1-\frac{1}{e}}\br{1-e^{-y}}\cdot \frac{1}{2}x_{ab'}\cdot p_{ab'}}e^{-y}\mbox{d}y\\
 & \geq & \int_{0}^{\infty}\prod_{b'\in B\setminus b}\exp\br{-\br{1-\frac{1}{e}}\frac{1}{2}x_{ab'}\cdot p_{ab'}\cdot y}e^{-y}\mbox{d}y\\
 & = & \frac{1}{1+\br{1-\frac{1}{e}}\frac{1}{2}\br{\sum_{b'\in B\setminus b}p_{ab'}\cdot x_{ab'}}} \\
 & \geq & \frac{1}{1+\frac{1}{2}\br{1-\frac{1}{e}}} \\
 & = & \frac{2e}{3e-1}.
\end{eqnarray*}
In the first inequality above, we used the fact that $1-c(1-e^{-y})\geq e^{-cy}$
for $c\in[0,1]$ and any $y\in\mathbb{R}$: here $c=\br{1-\frac{1}{e}}\frac{1}{2}x_{ab'}\cdot p_{ab'}$.
In the first equality we used $\int_{0}^{\infty}e^{-ax}\mbox{d}x=\frac{1}{a}$.
In the last inequality we used the LP constraint $\sum_{b'\in B\setminus b}p_{ab'}\cdot x_{ab'}\leq1$.

Altogether, as anticipated earlier, 
\[
\pr{ab\mbox{ is probed}}\geq\br{1-\frac{1}{e}}\frac{x_{ab}}{2}\cdot\frac{2e}{3e-1}=x_{ab}\cdot\frac{e-1}{3e-1}>0.24\cdot x_{ab}.
\]

\paragraph{Technical details.} Recall that we assumed that we are able to compute the quantities $\beta_{ab}$, hence the desired dumping factors $\alpha_{ab}$. Indeed, for our goals it is sufficient to estimate them with large enough probability and with sufficiently good accuracy. This can be done by simulating the underlying random process a polynomial number of times. This way the above probability can be lower bounded by $(\frac{e-1}{3e-1}+\eps) x_e$ for an arbitrarily small constant $\eps>0$. In particular, by choosing
a small enough $\eps$ the factor $0.245$ is still guaranteed. The details are given in Appendix~\ref{sec:computeDumping}.

The approximation factor can be further improved to $0.245$ via the technique based on small and big probabilities that we introduced before. This is discussed in the next section. Theorem \ref{thr:mainOnline} follows.

%The remaining technical details are given in Appendix Section~\ref{sec:computeDumping}.
%%The omitted technical details will be given in the \fab{full version} of the paper (\fab{see also \cite{AGM15arxiv}}).\fabr{Added citation to arxiv version} 
%Theorem \ref{thr:mainOnline} follows.

%dummy comment inserted by tex2lyx to ensure that this paragraph is not empty
%dummy comment inserted by tex2lyx to ensure that this paragraph is not empty

\subsection{Combination with Greedy in the Online Case}
\label{sec:onlineBigSmall}

%\fabr{Let's use g(p) and h(p) consistently. We might also use g(x,p) in two variables...}

Recall that $h(p)=\frac{1}{1+p}\br{1-\exp\br{-\br{1+p}\frac{1}{p}\ln\frac{1}{1-p}}}$.
We are again applying the big/small probabilities trick, so let
$\delta\in(0,1)$ be a parameter to be fixed later. Consider again the subroutine for buyers.
We previously used dumping factors $\alpha_{ab}=\frac{1}{2\beta_{ab}}$,
where we had set $\beta_{ab}\geq h\br{p_{ab}}$.
%\[
%\beta_{ab}\geq\frac{1}{1+p_{ab}}\br{1-\exp\br{-\br{1+p_{ab}}\frac{1}{p_{ab}}\ln\frac{1}{1-p_{ab}}}}=\fab{h}\br{p_{ab}}.
%\]

This time we define $\alpha_{ab}=\frac{1}{\beta_{ab}}h\br{\delta}$
for $ab$ such that $p_{ab}\leq\delta$, and $\alpha_{ab}=\frac{1}{\beta_{ab}}\frac{1}{2}$
otherwise. We again assume here that we can calculate $\beta_{ab}$
(see Section~\ref{sec:computeDumping}). Define $E_{large}=\setst{ab\in E}{p_{ab}\geq\delta}$
and $E_{small}=E\setminus E_{large}$, and let $LP_{large}=\gamma\cdot LP_{onl}$.
Therefore, for edge $ab$ the probability that $ab$ is probed when
$b$ scans items is exactly $h\br{\delta}$ for $ab\in E_{small}$
and $\frac{1}{2}$ for $ab\in E_{large}.$ Now by repeating the steps
in the proof of Section \ref{sec:online}, we obtain that the probability that $ab$ is
not blocked on $a$'s side is at least 
\begin{align*}
\frac{1}{1+\br{1-\frac{1}{e}}\br{\sum_{b'\in B\setminus b}p_{ab'}\cdot\alpha_{ab'}\beta_{ab'}\cdot x_{ab'}}}\geq & \frac{1}{1+\br{1-\frac{1}{e}}h\br{\delta}\br{\sum_{b'\in B\setminus b}p_{ab'}\cdot x_{ab'}}}\\
\geq & \frac{1}{1+\br{1-\frac{1}{e}}h\br{\delta}},
\end{align*}
since $\alpha_{ab'}\cdot\beta_{ab'}=h\br{\delta}$ for small
edges and $\alpha_{ab'}\cdot\beta_{ab'}=\frac{1}{2}\leq h\br{\delta}$
for large edges. Therefore, the approximation ratio of such an algorithm
is at least 
\begin{multline*}
\br{1-\frac{1}{e}}\br{\gamma\frac{1/2}{1+h\br{\delta}\br{1-\frac{1}{e}}}+\br{1-\gamma}\frac{h(\delta)}{1+h\br{\delta}\br{1-\frac{1}{e}}}}\\
=\br{1-\frac{1}{e}}\frac{1}{1+h\br{\delta}\br{1-\frac{1}{e}}}\br{\gamma\frac{1}{2}+\br{1-\gamma}h\br{\delta}}.
\end{multline*}

An alternative algorithm simply computes a maximum weight matching
w.r.t. weights $p_{e}w_{e}$ in the graph corresponding to LP-ONL,
and upon arrival of the first copy of a buyer type $b$ probes only the edge incident to
$b$ in the matching (if any). By the same argument as in the offline
case, this matching has weight at least $\gamma\cdot\delta\cdot LP_{onl}$,
and every buyer type is sampled with probability at least $1-\frac{1}{e}$.
So, the approximation ratio of the greedy algorithm is at least $\br{1-\frac{1}{e}}\gamma\delta$.

%In the online case we can also use the greedy algorithm that computes
%the maximum matching. We compute a matching of value at least $\gamma\cdot\delta\cdot LP_{onl}$,
%and upon arrival of a buyer we probe only the edge to which \fab{she} is
%adjacent (if any). Every node is realized with probability at least
%$1-\frac{1}{e}$, so the approximation ratio of such an algorithm
%is at least $\br{1-\frac{1}{e}}\gamma\delta$.

For a fixed $\delta$, depending on the value of $\gamma$ (that
we can compute offline) we can run the algorithm with best approximation
ratio according to the above analysis. Thus the overall approximation
ratio is 
\[
(1-\frac{1}{e})\max\left\{ \frac{1}{1+h\br{\delta}\br{1-\frac{1}{e}}}\br{\gamma\frac{1}{2}+\br{1-\gamma}h\br{\delta}},\gamma\cdot\delta\right\} .
\]
%Again \fab{in the worst case} $\gamma$ satisfies 
%\[
%\frac{1}{1+\fab{h}\br{\delta}\br{1-\frac{1}{e}}}\br{\gamma\frac{1}{2}+\br{1-\gamma}\fab{h}\br{\delta}}=\gamma\cdot\delta,
%\]
%which gives 
%\[
%\gamma=\br{\br{\delta\cdot \fab{h}\br{\delta}\br{1-\frac{1}{e}}+\delta-\frac{1}{2}}\frac{1}{\fab{h}\br{\delta}}+1}^{-1}.
%\]
As in Section~\ref{sub:Bipartite-graphs} the worst-case is obtained when the two quantities are equal. This yields
\[
\gamma = \frac{h\br{\delta}}{\delta \cdot \br{1+h\br{\delta}\br{1-\frac{1}{e}}} -\frac{1}{2} +h\br{\delta}}.
\]
Hence the actual approximation ratio is
\[
\br{1-\frac{1}{e}}\frac{\delta \cdot h\br{\delta}}{\delta \cdot \br{1+h\br{\delta}\br{1-\frac{1}{e}}} -\frac{1}{2} +h\br{\delta}},
\] 
and now we just need to optimize over $\delta$.
When we set $\delta=0.74$, the approximation ratio becomes approximately 
$0.245712219628$ -- which allows us to claim the bound of $0.245$ from Theorem~\ref{thr:mainOnline}.

\bibliographystyle{acm}
\bibliography{StochasticMatching.bib,refs.bib}

\appendix
\section{Computing Dumping Factors}\label{sec:computeDumping}

Recall that we assumed the knowledge of quantities $\beta_{ab}$,
which are needed to define the dumping factors $\alpha_{ab}$. Though
we are not able to compute the first quantities exactly in polynomial
time, we can efficiently estimate them and this is sufficient for
our goals. Let us focus on a given edge $ab$. Recall that 
\begin{multline*}
\beta_{ab}:=\excondls{\h A{}_{b}\setminus a,Y}{\prod_{a'\in A_{b}:Y_{a'b}<Y_{ab}}\br{1-p_{a'b}}}{a\in\h A_{b}}\\
\geq\frac{1}{1+p_{ab}}\br{1-\exp\br{-\br{1+p_{ab}}\frac{1}{p_{ab}}\ln\frac{1}{1-p_{ab}}}}=h\br{p_{ab}}.
\end{multline*}

Let us simulate the subroutine for buyers $N$ times without the dumping
factors: in a simulation we run GKPS and sample the $Y$ variables,
but \emph{simulate} the probes of the edges without actually probing any edge.
We shall set $N$ later. Let $S^{1},S^{2},...,S^{N}$ be $0$-$1$
indicator random variables of whether $a$ was safe or not in each
simulation. Note that $\ex{S^{i}}=\beta_{ab}x_{ab}\in\brq{h\br{p_{ab}}x_{ab},x_{ab}}$.

Suppose that $x_{ab}\geq\frac{\epsilon}{n}$, where $n$ is the number
of buyers. The expression $\hat{s}_{ab}=\frac{1}{N}\sum_{i=1}^{N}S^{i}$
should be a good estimation of $\beta_{ab}\cdot x_{ab}$, i.e., \[
\hat{s}_{ab}\in\brq{\beta_{ab}x_{ab}\br{1-\epsilon},\beta_{ab}x_{ab}\br{1+\epsilon}}\]
with probability $1-\frac{1}{n^{C}}$. \global\long\def\Z{Z}
 Set $N=\frac{6n}{\epsilon^{3}}\ln\br{2n^{2}\Z}$ for $Z = \frac{3}{\eps}+1$.

Applying Chernoff's bound $\pr{|X-\ex X|>\eps\ex X}\leq2e^{-\frac{\epsilon^{2}}{3}\ex X}$
with $X=\sum_{i=1}^{N}S_{i}$ one obtains: %From Chernoff's bound with $X=\sum_{i=1}^{N}\fab{S}_{i}$: 
%\[
%\pr{X>\br{1+\epsilon}\ex X}\leq\exp\br{-\frac{\epsilon^{2}}{3}\ex X}\mbox{ and }\pr{X<\br{1-\epsilon}\ex X}\leq\exp\br{-\frac{\epsilon^{2}}{3}\ex X},
%\]
%so 
\begin{align*}
 & \pr{\sum_{i=1}^{N}S_{i}\notin\brq{\br{1-\epsilon}\beta_{ab}x_{ab}\cdot N,\br{1+\epsilon}\beta_{ab}x_{ab}\cdot N}}\\
\leq & 2\exp\br{-\frac{\epsilon^{2}}{3}\beta_{ab}x_{ab}\cdot N}\leq2\exp\br{-\frac{\epsilon^{2}}{3}\frac{x_{ab}}{2}\cdot N}\leq2\exp\br{-\frac{\epsilon^{3}}{6n}\cdot N}=\frac{1}{n^{2}}\frac{1}{\Z}.
\end{align*}

From the union bound, with probability at least $1-\frac{1}{\Z}$
we have that $\h s_{ab}\in\brq{\beta_{ab}x_{ab}\br{1-\epsilon},\beta_{ab}x_{ab}\br{1+\epsilon}}$
for every edge $ab$ such that $x_{ab}\geq\frac{\epsilon}{n}$.
Now let us assume that this happened, i.e., that we have good estimates. We
set $\alpha_{ab}=\max\{\frac{1}{2},\min\{\frac{1}{2}\frac{x_{ab}}{\h s_{ab}},1\}\}$
which belongs to $\brq{\frac{1}{2}\frac{1}{\beta_{ab}\br{1+\epsilon}},\frac{1}{2}\frac{1}{\beta_{ab}\br{1-\epsilon}}}$,
but only for edges $ab$ such that $x_{ab}\geq\frac{\epsilon}{n}$.
%;
%\mar{the interval may contain $1$, but this is not a problem}.
For edges $ab$ such that $x_{ab}<\frac{\epsilon}{n}$ we just put
$\alpha_{ab}=1$ (so we do not dump such edges actually). Two elements
of the proof were depending on the dumping factors. First, now the
probability that edge $ab$ is taken is $\alpha_{ab}\beta_{ab}x_{ab}\in\brq{\frac{x_{ab}}{2\br{1+\epsilon}},\frac{x_{ab}}{2\br{1-\epsilon}}}$.
Second, recall that the probability of edge $ab$ not to be blocked is:
\begin{equation}
\frac{1}{1+\br{1-\frac{1}{e}}\br{\sum_{b'\in B\setminus b}p_{ab'}\cdot\alpha_{ab'}\beta_{ab'}\cdot x_{ab'}}}.\label{eq:itemblockedwithdumping}
\end{equation}
We have that 
\begin{align*}
 & \sum_{b'\in B\setminus b}p_{ab'}\cdot\alpha_{ab'}\beta_{ab'}\cdot x_{ab'}\\
= & \sum_{b'\in B\setminus b:x_{ab'}\geq\frac{\epsilon}{n}}p_{ab'}\cdot\alpha_{ab'}\beta_{ab'}\cdot x_{ab'}+\sum_{b'\in B\setminus b:x_{ab'}<\frac{\epsilon}{n}}p_{ab'}\cdot\alpha_{ab'}\beta_{ab'}\cdot x_{ab'}\\
\leq & \sum_{b'\in B\setminus b:x_{ab'}\geq\frac{\epsilon}{n}}p_{ab'}\cdot\frac{1}{2\br{1-\epsilon}}x_{ab'}+\sum_{b'\in B\setminus b:x_{ab'}<\frac{\epsilon}{n}}x_{ab'}\\
\leq & \frac{1}{2\br{1-\epsilon}}+\epsilon=\frac{1}{2}+O\br{\epsilon}.
\end{align*}
So the probability that $a$ is not blocked is at least $\frac{1}{1+\br{1-\frac{1}{e}}\br{\frac{1}{2}+O\br{\epsilon}}}.$
The final probability that edge $ab$ is probed is at least 
\begin{align*}
\br{1-\frac{1}{e}}\frac{x_{ab}}{2\br{1+\epsilon}}\cdot\frac{1}{1+\br{1-\frac{1}{e}}\br{\frac{1}{2}+O\br{\epsilon}}} & =\frac{x_{ab}}{1+\eps}\cdot\frac{e-1}{2e+\br{e-1}\br{1+O\br{\epsilon}}}\\
 & =x_{ab}\cdot\frac{e-1}{3e-1+O\br{\epsilon}}>0.24 \cdot x_{ab}.
\end{align*}
In the last inequality above we assumed $\eps$ to be small enough.

With probability at most $\frac{1}{\Z}$ we did not obtain good estimates
of the dumping factors. Still we have that $\alpha_{ab}\in\brq{\frac{1}{2},1}$,
and therefore $\alpha_{ab}\beta_{ab}\in\brq{\frac{1}{4},1}$. In this
case quantity~\eqref{eq:itemblockedwithdumping} can be just lower-bounded
by $\frac{1}{1+\br{1-\frac{1}{e}}}$, and the probability that edge
$ab$ is probed in the subroutine for buyers is at least $\frac{x_{ab}}{4}$.
Thus the probability that edge $ab$ is probed during the
algorithm is at least $\br{1-\frac{1}{e}}\frac{x_{ab}}{4}\cdot\frac{1}{1+\br{1-\frac{1}{e}}}=\frac{x_{ab}}{4}\cdot\frac{e-1}{2e-1}> 0.097 x_{ab}.$
The total expected outcome of the algorithm is therefore, for sufficiently
small $\eps$, at least 
\[
LP_{onl}\br{\br{1-\frac{1}{\Z}}\frac{e-1}{3e-1+O\br{\epsilon}}+\frac{1}{\Z}\frac{1}{4}\cdot\frac{e-1}{2e-1}}\overset{Z = \frac{3}{\eps}+1}{\geq} 0.24 \cdot LP_{onl}.
\]

%***
%Total expected outcome of our algorithm is 
%\begin{align*}
% & \pr{\mbox{GOOD estimations}}\excond{ALG}{\mbox{GOOD estimations}}\\
% & +\pr{\mbox{BAD estimations}}\excond{ALG}{\mbox{BAD estimations}}\\
%\geq & \pr{\mbox{GOOD estimations}}\frac{1}{4.16}LP_{onl}+\pr{\mbox{BAD estimations}}\frac{1}{10.33}LP_{onl}\\
%\geq & LP_{onl}\br{\br{1-\frac{1}{\Z}}\frac{1}{4.16}+\frac{1}{\Z}\frac{1}{10.33}}\\
%\geq & LP_{onl}\br{\frac{1}{4.16}-\frac{1}{\Z}\frac{1}{6.96}},
%\end{align*}
%and since $Z=3\frac{1}{\eps}+1$, the last quantity is at least $\frac{1}{4.16+\eps}\cdot LP_{onl}$.

The above approach can be combined with the small/big probability
trick from Section \ref{sec:onlineBigSmall}. %\fab{One just needs to replace $\beta_{ab}$ with $\frac{\hat{s}_{ab}}{x_{ab}}$.} 
By choosing $\eps$ small enough the approximation ratio is $0.245$
as claimed.

\end{document}